\documentclass[12pt,reqno]{amsart}

\newtheorem{theorem}{{\sc Theorem}}[section]

\newtheorem{lemma}[theorem]{{\sc Lemma}}
\newtheorem{proposition}[theorem]{{\sc Proposition}}

\theoremstyle{remark}
\newtheorem{remark}[theorem]{{\sc Remark}}

\begin{document}

\title[worst case value at risk]
{A duality approach to the worst case value at risk for a sum of dependent random variables 
with known covariances}
\author{Brice Franke}
\address{Ruhr-Universit\"at Bochum, Fakult\"at f\"ur Mathematik, NA 3/35, D-44780 Bochum, Germany}
\email{brice.franke@ruhr-uni-bochum.de}

\author{Michael Stolz}
\address{Ruhr-Universit\"at Bochum, Fakult\"at f\"ur Mathematik, NA 3/69, D-44780 Bochum, Germany. 
(Current address: Universit\"at Duisburg-Essen, Campus Duisburg, LE 309a, D-47057 Duisburg, Germany)}
\email{michael.stolz@ruhr-uni-bochum.de}

\date{\today}

\begin{abstract} We propose an approach to the aggregation of risks which is based on estimation of simple quantities (such as covariances) associated to a vector of dependent random variables, and which avoids the use of parametric families of copulae. Our main result demonstrates that the method leads to bounds on the worst case Value at Risk for a sum of dependent random variables. Its proof applies duality theory for infinite dimensional linear programs. 
\end{abstract}

\maketitle

\section{Introduction}

Aggregation of risks is a key issue for risk management and regulation in the financial sector. 
Consider an institution which may incur losses $X_1, \ldots, X_n$ in each of its $n$ different
 business divisions. To assess its overall exposure to risk, knowledge of the distribution of the sum
 $S = X_1 + \ldots + X_n$ is vital. Suppose that the distributions of the individual $X_i$ are reasonably
 well-known. This does not imply, however, that there are enough data available to estimate the
 distribution of $S$, let alone the joint distribution of the vector $X = (X_1, \ldots, X_n)$.
 As the process of implementing the Basel II and Solvency II accords has amply demonstrated,
 the data on which the knowledge of the marginal distributions is based are often not suitable for reconstructing 
the joint distribution: They may have been collected at
 different times, asynchronously, or subject to wildly differing standards of exactness or
 documentation (see, e.g., \cite{Romeike-etal}). Nonetheless, it is often the case that at least 
{\it some} data from the joint distribution are
 available --- perhaps not enough to obtain a reasonably accurate picture of the tail behavior
 of $S$, but sufficient to estimate quantities such as the covariances of the $X_i$.\\
Suppose now, for simplicity, that one is ultimately interested in the Value at Risk of $S$
(at some level $\alpha$). If there were no information whatsoever available about the dependence
structure of $X$, one would have to resort to a worst case analysis in the sense that one would let 
the joint distribution of $X$ run through all $n$-variate distributions with the appropriate
marginals. In the situation
at hand, with at least some aspects of the dependence structure known,
it is natural to take only those joint distributions into account that are compatible with the additional information.\\
Arguably, this procedure makes sense even in the case when one has enough data from the joint
distribution to feel comfortable with estimating the distributions of $S$ and even of $X$.
In this situation, the worst case analysis subject to a constraint might be viewed as a robust
version of VaR estimation (robust in the sense of taking model uncertainty into account,
see \cite{FollmerSchiedBuch}) --- even more so as the constraints we are concerned with are stated using 
quantities that are amenable to robust estimation (in the statistician's sense of the word).\\

The present paper treats the ``worst case value at risk''
of a sum of two  random variables whose covariance is assumed to be known. We rephrase the worst
case analysis as an infinite dimensional linear program and establish a handy dual version, from which
it is easy to extract bounds for the worst case value at risk. Making the tools of duality theory 
bear upon the present setup
constitutes the technical core of what follows.\\

Before providing, in Section \ref{technintro}, an exact statement of the main result, let us try
to place our approach within the context of risk management literature and practice. The last decade has seen 
an upsurge of interest in the concept of copulae. Parametric families of copulae have been widely used as a basis
for estimation of various quantities related to aggregate risk (see, e.g., \cite{McNeilFreyEmbrechts} and the references therein), even as these procedures have also met fundamental criticism (see \cite{MikoschExtremes2006}). By contrast, the 
methods which are advertised in the present paper avoid any use of parametric families -- as well as of inversion of marginals -- and are based on estimation of scalar or finite-dimensional quantities, not of real valued functions on a unit cube.

Among the literature that is based on the copula concept, it is a paper of Embrechts, H\"oing, and Juri 
\cite{EmbrechtsHoingJuri} that is closest in spirit to the present approach. The authors  
use order relations for copulae to exploit partial 
information on the dependence structure of $X$ for worst case bounds on $\operatorname{VaR}_{\alpha}(S)$. Specifically,
if $C$ denotes the copula of $X$ and if there are copulae $C_0, C_1$ such that $C_0 \le C$ and $C^d \le C_1^d$ ($C^d$ denoting the dual copula), then the authors find upper and lower bounds for $\operatorname{VaR}_{\alpha}(S)$. The 
copulae $C_0, C_1$
may be chosen to encode properties of $X$ such as positive orthant dependence. In contrast to that, the present 
approach is more empirical in spirit, in that the partial information we use can be easily estimated from the data available. 

\section{Worst case analysis based on marginals and covariance}
\label{technintro}

Suppose that $X$ and $Y$ represent losses associated with different positions, and that a risk manager
wishes to calculate $\operatorname{VaR}_{\alpha}(S)$, the Value at Risk at the level $\alpha$, 
of $S = X + Y$, that is, the infimum 
over all $z \in {\mathbb R}$ such that
$$ {\mathbb P}(S \le z) = \iint 1_{(- \infty, z)}(x+y)\ \mu(dx, dy)\ \ge \alpha,$$
where $\mu$ denotes the joint distribution of the vector $(X, Y)$. Suppose now that
nothing about $\mu$ is known, except for that its marginals are $P$, the distribution of $X$, and $Q$, the distribution of $Y$, respectively. Then it seems natural to replace $\operatorname{VaR}_{\alpha}(S)$ 
with $\operatorname{WVaR}_{\alpha}(P, Q)$, which we define as
$$  \operatorname{WVaR}_{\alpha}(P, Q) := \inf \left\{ z \in {\mathbb R}:\ \inf_{\mu \in {\mathcal M}(P, Q)} \iint  1_{(- \infty, z)}(x+y)\ \mu(dx, dy)\ \ge \alpha\right\},$$ 
with 
$$ {\mathcal M}(P,Q):=\left\{\mu\in{\mathcal M}^1(\mathbb{R}^2):\mu\circ{\rm pr}_1^{-1} = P,\
             \mu\circ{\rm pr}_2^{-1} = Q \right\}. $$ 
Here $ {\rm pr}_1 $ and $ {\rm pr}_2 $  are the projections on the components of 
$ \mathbb{R}\times\mathbb{R} $, and $ {\mathcal M}^1(\mathbb{R}^2) $ denotes the set of probability measures on 
$ \mathbb{R}^2 $. 
Note that $\operatorname{WVaR}_{\alpha}(P, Q) \ge \operatorname{VaR}_{\alpha}(S)$. It is thus natural to regard
$\operatorname{WVaR}_{\alpha}(P, Q)$ as a ``worst case Value at Risk'', taking care of the model uncertainty 
inherent in estimating quantiles of the distribution of a sum of dependent random variables whose joint 
distribution is not (completely) known. \\
Now suppose that {\it some} (though not necessarily exceedingly copious) data from the joint distribution of $(X, Y)$
are available. These can be used to estimate, e.g., the covariance  $$  \iint xy\ \mu(dx,dy)-\int x P(dx)\int y Q(dy), $$   
which typically can be estimated in a robust way on the basis of rather few data -- in stark contrast to the wealth
of data that would be necessary for a reliable estimate of the copula of $(X, Y)$. Incorporating this
extra information into the worst case analysis will lead to a restricted worst case Value at Risk that will in general 
be smaller than $\operatorname{WVaR}_{\alpha}(P, Q)$. To be specific, set
$$ \operatorname{WVaR}_{\alpha}(P, Q, k) := \inf\left\{ z \in {\mathbb R}:\ 
\inf_{\mu \in {\mathcal M}(P, Q, k)} \iint  1_{(- \infty, z)}(x+y)\ \mu(dx, dy)\ \ge \alpha\right\},$$
where $${\mathcal M}(P,Q, k):=\left\{\mu\in{\mathcal M}(P,Q):\iint xy\ \mu(dx,dy)= k \right\}.$$ 
Note that the obvious choice for $k$ is $k = \hat{\rho}_{X,Y}+e_Pe_Q$,
where $\hat{\rho}_{X,Y}$ is an estimate of the covariance of $X$ and $Y$, and
$$  e_P:=\int xP(dx)\quad {\rm and} \quad e_Q:=\int yQ(dy)$$
are assumed to be known.\\

The crucial problem, which will be treated in what follows, is how to evaluate an infimum  
\begin{equation}
\label{transport}
 \inf_{\mu \in {\mathcal M}(P, Q, k)} \iint  \ell(x+y)\ \mu(dx, dy),
\end{equation}
where $\ell(x, y)$ is a lower semicontinuous bounded function such as $1_{(- \infty, z)}(x+y)$. (This 
level of generality should 
provide some leeway for discussions of more general aggregate financial positions (see \cite[p.248]{McNeilFreyEmbrechts})
and risk measures.) In the absence of extra information about the joint distribution, i.e., if the 
infimum in \eqref{transport} is taken over ${\mathcal M}(P, Q)$ rather than ${\mathcal M}(P, Q, k)$, the problem
is but an instance of the classical (Monge-Kantorovich) mass transport problem: A substance 
whose initial spacial distribution is given by $P$, is to be shipped to a final spacial distribution given by $Q$, according to a transportation plan which can be encoded by an element of ${\mathcal M}(P, Q)$ and which 
is supposed to minimize the overall cost if transport from $x$ to $y$ comes at the price $\ell(x, y)$. 
Monographic treatments of mass transport are due to Rachev and R\"uschendorf \cite{RaRu1, RaRu2} and Villani 
\cite{VillaniOldNew}. The special case of $\ell(x, y) = 1_{(- \infty, z)}(x+y)$ has been studied in  
Makarov \cite{Makarov1981-engl}, R\"uschendorf \cite{RuschendorfAdvApplProb1982},
and Frank, Nelsen, and Schweizer \cite{FrankNelsenSchweizer}.\\ 

The modern functional-analytic approach to mass transport, which is due to
Kantorovich, is based on the following duality:
\begin{equation}  
\label{mtdual}
\inf_{\mu\in{\mathcal M}(P,Q)}\iint\ell(x,y)\mu(dx,dy)
          =\sup_{(f,g)\in {\mathcal N}(\ell)}\left(\int f(x)P(dx)+\int g(y)Q(dy)\right),
\end{equation}
where
$$  {\mathcal N}(\ell):=\left\{(f,g)\in C_b(\mathbb{R})\times C_b(\mathbb{R}):f(x)+g(y)\leq \ell(x,y)\ \forall x, y \in {\mathbb R} \right\},$$
$C_b(\mathbb{R})$ denoting the set of bounded continuous functions on ${\mathbb R}$.
Note that by plugging suitable test functions into the right-hand side, one obtains lower bounds on the left-hand side.  
Returning to $\operatorname{WVaR}_{\alpha}(P, Q)$, one observes that replacing for each $z$ and 
$\ell(x, y) = 1_{(- \infty, z)}(x+y)$ the left-hand side of \eqref{mtdual} by a lower bound will lead to an upper bound 
on $\operatorname{WVaR}_{\alpha}(P, Q)$, and this is what matters in risk management. (Of course, a lower bound on 
$\operatorname{WVaR}_{\alpha}(P, Q)$ can be obtained by plugging, for each $z$, test measures into the left-hand side.)\\

Given the usefulness of \eqref{mtdual}, one would like to have an analogous statement for the case
that the infimum is taken over ${\mathcal M}(P, Q, k)$ rather than ${\mathcal M}(P, Q)$. This is the content 
of the main result of this paper, which will be proven in Section \ref{beweis}:

\begin{theorem}
\label{hauptsatz}
Assume that $ \ell $ is bounded and lower semicontinuous and that $P, Q \in {\mathcal M}^1({\mathbb R})$ satisfy
$$  \int x^2 \ P(dx)<\infty,\ \ \int y^2 \ Q(dy)<\infty .$$
Then for all $ k\in\mathbb{R} $ with
\begin{equation}
\label{rangek}
   \inf_{\mu\in{\mathcal M}(P,Q)}\iint xy\ \mu(dx,dy) \leq k\leq   \sup_{\mu\in{\mathcal M}(P,Q)}\iint xy\ \mu(dx,dy),
\end{equation}
one has that
\begin{equation}
\label{mainduality}  
\inf_{\mu\in{\mathcal M}(P,Q,k)}\iint \ell(x,y)\ \mu(dx,dy) = \sup_{(f,g,\alpha)\in {\mathcal U}(\ell)}
   \Bigg(\int f(x)P(dx)+\int g(y)Q(dy)+\alpha k\Bigg) ,
\end{equation}
where 
$$   {\mathcal U}(\ell)
  :=\left\{(f,g,\alpha)\in S(\mathbb{R})\times S(\mathbb{R})\times\mathbb{R}:f(x)+g(y)
+\alpha xy\leq\ell(x,y)\ \forall x, y \in {\mathbb R} \right\} $$
and
$$  S(\mathbb{R}):=\left\{f\in C(\mathbb{R}):f(x)\big/(x^2\vee1)\ \text{\rm is bounded} \right\}    .$$
Moreover, a minimizing measure $ \mu\in{\mathcal M}(P,Q,k) $ for the left-hand side of \eqref{mainduality} exists.
\end{theorem}

\vspace{1em}

\begin{remark}
~
\begin{itemize}
\item[(1)] We note that if the supports of the measures $ P $ and $ Q $ are  bounded, then in the definition of $ {\mathcal U}(\ell) $ one can replace $S(\mathbb{R}) \times S(\mathbb{R})$ by $ C_b(\operatorname{supp}(P)) 
\times C_b(\operatorname{supp}(Q)) $. 
\item[(2)] It is not clear whether there exist maximizing vectors $ (f,g,\alpha) $ for the maximization 
problem on the right-hand side of \eqref{mainduality}. In order to obtain such kind of result one 
usually needs apriori bounds on $ f $, $ g $ and $ \alpha $. 
If $ \alpha $ is bounded, apriori bounds can be obtained by replacing $ f $ and $ g $ by suitable 
new functions $ f_0 $ and $ g_0 $ (see, e.g., \cite{JimenezRodriguez}). 
However, in our situation no apriori bound for $ \alpha $ is available.
\item[(3)] To obtain explicit values for the bounds on $k$ in \eqref{rangek}, 
note that they are solutions to classical 
transportation problems with cost function $ c:\mathbb{R}^2\rightarrow\mathbb{R};(x,y)\mapsto xy $. 
The cost function $ c $ is twice continuously differentiable and the mixed partial 
derivatives $ c_{xy} $ are positive.
Therefore there exists an optimal transportation plan for those optimization problems; 
i.e.: there exist kernels $ \Gamma^+(x,dy) $ and $ \Gamma^-(x,dy) $ with the property that the measures
$ \mu^+(dx,dy):=\Gamma^+(x,dy)P(dx) $ and $ \mu^-(dx,dy):=\Gamma^-(x,dy)P(dx) $ satisfy
$$
   \inf_{\mu\in{\mathcal M}(P,Q)}\iint xy\ \mu(dx,dy)=\int xy\ \mu^-(dx,dy) $$ and 
$$          \sup_{\mu\in{\mathcal M}(P,Q)}\iint xy\ \mu(dx,dy)=\int xy\ \mu^+(dx,dy).
$$
In order to provide a more explicit representation for those kernels in terms of the distribution 
functions $ F $ and $ G $ of the measures $ P $ resp. $ Q $, we introduce the generalized 
inverse of a non-decreasing right-continuous function $ H:\mathbb{R}\rightarrow\mathbb{R} $ as 
$$  H^-(\xi):=\inf\{x\in\mathbb{R}^+:F(x)\geq\xi\} ,$$
the left-limit of $ H $ at the point $ x\in\mathbb{R} $ as 
$ H(x-):=\lim_{\epsilon\downarrow0}H(x-\epsilon) $ and the jump of $ H $ at $ x\in\mathbb{R} $ as
$ \Delta H(x):=H(x)-H(x-) $.
Further, we use the notation $ \delta_u(dv) $ to denote the Dirac measure on $ \mathbb{R} $, which 
assigns a unit point mass to a single point $ u\in\mathbb{R} $. Moreover, for a Borel-measurable set
$ A\subset\mathbb{R} $ with positive Lebesgue measure we denote by $ U_A(dy) $ the uniform 
distribution on $ A $. 
We then have 
$$ \Gamma_+(x,dy):=\left\{\begin{array}{cc} 
    \delta_{G^-(F(x))}(dy) &  {\rm if} \ \Delta F(x)=0 \\
    U_{[G^-(F(x-)),G^-(F(x))]}(dy) & {\rm if} \ \Delta F(x)\neq 0 \end{array}\right. $$ 
and 
$$ \Gamma_-(x,dy):=\left\{\begin{array}{cc} 
    \delta_{G^-(1-F(x))}(dy) &  {\rm if} \ \Delta F(x)=0 \\
    U_{[G^-(1-F(x-)),G^-(1-F(x))]}(dy) & {\rm if} \ \Delta F(x)\neq 0 \end{array}\right. .$$
This explicit representation follows from the fact that the supports of the 
optimal measures satisfy strong constraints (see \cite[p.96]{AndersonNashBuch}), which restricts 
the set of possible transportation plans to marginal transformations, which transform the measure $ P $ 
to the measure $ Q $.
\end{itemize}
\end{remark}

In the case that $X$ or $Y$ has a heavy-tailed distribution, Theorem \ref{hauptsatz} is not applicable as stated.
In this situation, one may replace $xy$ by $(xy \vee -R) \wedge R$ for a suitable $R$, a natural choice from the 
point of view of robust statistics. Then the following variant of Theorem \ref{hauptsatz} applies, whose proof is
actually a simplified version of the one given in Section \ref{beweis} for Theorem \ref{hauptsatz} and will therefore
be omitted. For $\kappa \in \operatorname{C}_b({\mathbb R}^2)$ write
 $${\mathcal M}_{\kappa}(P,Q, k):=\left\{\mu\in{\mathcal M}(P,Q):\iint \kappa(x, y)\ \mu(dx,dy)= k \right\}.$$ 

\begin{theorem}
\label{hauptsatzvar}
Assume that $ \ell $ is bounded and lower semicontinuous and $\kappa \in \operatorname{C}_b({\mathbb R}^2)$.
Then for all $ k\in\mathbb{R} $ with
$$
   \inf_{\mu\in{\mathcal M}(P,Q)}\iint \kappa(x, y)\ \mu(dx,dy) \leq k\leq   \sup_{\mu\in{\mathcal M}(P,Q)}
\iint \kappa(x, y)\ \mu(dx,dy),
$$
one has that
\begin{equation}
\label{maindualityvar}  
\inf_{\mu\in{\mathcal M}_{\kappa}(P,Q,k)}\iint \ell(x,y)\ \mu(dx,dy) = \sup_{(f,g,\alpha)\in {\mathcal U}(\ell)}
   \left(\int f(x)P(dx)+\int g(y)Q(dy)+\alpha k\right) ,
\end{equation}
where 
$$   {\mathcal U}_{\kappa}(\ell)
  :=\left\{(f,g,\alpha)\in \operatorname{C}_b({\mathbb R}^2) \times\mathbb{R}:f(x)+g(y)
+\alpha \kappa(x, y) \leq\ell(x,y)\ \forall x, y \in {\mathbb R} \right\}. $$
Moreover, a minimizing measure $ \mu\in{\mathcal M}_{\kappa}(P,Q,k) $ for the left-hand side of \eqref{maindualityvar} exists.
\end{theorem}

\section{Proof of Theorem \ref{hauptsatz}}
\label{beweis}

In this section we will apply the theory of linear programming in infinite dimensional spaces
(see, e.g., Anderson and Nash \cite{AndersonNashBuch}) and will thus have to rephrase our problem 
in the framework of topological vector spaces. Denote by $ {\mathcal R}_b(\mathbb{R}^2) $
the space of signed Radon measures with finite total variation norm on $ \mathbb{R}^2 $. 
$ {\mathcal R}_b(\mathbb{R}^2) $ contains the cone  
$ {\mathcal R}_+(\mathbb{R}^2) $ of finite nonnegative Radon measures. 
By the Hahn decomposition theorem, every signed measure 
$ \mu\in{\mathcal R}_b(\mathbb{R}^2) $ can be decomposed into $ \mu=\mu^+-\mu^- $, where $ \mu^+ $ and $ \mu^- $ are elements of
$ {\mathcal R}_+(\mathbb{R}^2) $, and  one can define the total variation measure $ |\mu|:=\mu^++\mu^- $ 
(see, e.g., \cite{BauerMI}). 
Consider the vector spaces 
$$  \mathbb{X}:=\Bigg\{\mu\in {\mathcal R}_b(\mathbb{R}^2):\iint x^2+y^2 \ d|\mu|<\infty 
                  \Bigg\}  $$
and
$$   \mathbb{Y}:=  \Big\{h\in C(\mathbb{R}^2):h(x,y)\big/((x^2+y^2)\vee1)\in C_b(\mathbb{R}^2)\Big\} .$$ 
Note that $\mathbb{Y}$ contains the function $(x, y) \mapsto xy.$ The pairing
$$  \big\langle.,.\big\rangle:\mathbb{X}\times\mathbb{Y}
   \rightarrow\mathbb{R};\ (\mu, h)\mapsto\big\langle \mu, h \big\rangle:=\iint h(x,y) \mu(dx,dy)   $$
puts $\mathbb{X}$ and $\mathbb{Y}$ into duality. In $ \mathbb{X} $ one has the 
cone $ {\mathcal K}:=\mathbb{X}\cap{\mathcal R}_+(\mathbb{R}^2) $. 
We endow $ \mathbb{X} $ with the $ \sigma\big(\mathbb{X},\mathbb{Y}\big) $-topology,
which is the coarsest topology such that for all $  h\in \mathbb{Y} $ the functionals 
$ \mu\mapsto\langle \mu, h\rangle $ are continuous. Observe that $C_b({\mathbb R}^2) \subset {\mathbb Y}$,
so $\sigma\big(\mathbb{X},\mathbb{Y}\big)$-convergence implies weak convergence in the usual measure-theoretic
sense. (In what follows, ``weak concergence'' (written as $\stackrel{w}{\longrightarrow})$ will always
be understood in this sense. Otherwise, the relevant $\sigma$-topology will be made explicit.) We collect this and 
related useful properties in the following
\begin{lemma}
 \label{sigmakonv}
~ We write $\nu_k := ((x^2 + y^2) \vee 1)\mu_n,\  \nu := ((x^2 + y^2) \vee 1)\mu$.
\begin{itemize}
\item[(i)]$$\mu_n \stackrel{\sigma(\mathbb{X}, \mathbb{Y})}{\longrightarrow} \mu  \quad \Longleftrightarrow\quad
\nu_n \stackrel{w}{\longrightarrow} \nu.$$
\item[(ii)]  $$\mu_n \stackrel{\sigma(\mathbb{X}, \mathbb{Y})}{\longrightarrow} \mu  \quad \Longrightarrow\quad
\mu_n \stackrel{w}{\longrightarrow} \mu.$$
\item[(iii)] $$\mu_n \stackrel{\sigma(\mathbb{X}, \mathbb{Y})}{\longrightarrow} \mu  \quad \Longrightarrow\quad
(\nu_n)\ \text{{\rm  is tight}}.$$
\item[(iv)] $$(\nu_n)\ \text{{\rm  is tight}}\quad \Longrightarrow\quad (\mu_n)\ 
\text{{\rm  is tight}}.$$
\end{itemize}
\end{lemma}

We define the vector spaces 
$$   \mathbb{V}:=\Bigg\{\nu\in {\mathcal R}_b(\mathbb{R}):\int x^2 \ d|\nu|<\infty\Bigg\}
   \times 
 \Bigg\{\nu\in {\mathcal R}_b(\mathbb{R}):\int y^2 \ d|\nu|<\infty\Bigg\}\times\mathbb{R} $$
and 
$$    \mathbb{W}:= \Big\{f\in C(\mathbb{R}):f(x)\big/(x^2\vee1)\in C_b(\mathbb{R})\Big\}\times
        \Big\{g\in C(\mathbb{R}):g(y)\big/(y^2\vee1)\in C_b(\mathbb{R})\Big\}\times\mathbb{R} .$$
The vector spaces $ \mathbb{V} $ and $ \mathbb{W} $ form a dual pair with pairing
$$  \big\langle.,.\big\rangle:\mathbb{V}\times\mathbb{W}\rightarrow\mathbb{R};
     \Big((\mu_1,\mu_2,\alpha), (f,g,\beta) \Big)\mapsto \int f(x)\mu_1(dx)+\int g(y)\mu_2(dy)+\alpha\beta .$$
Define the constraint operator
$$ A:\mathbb{X}\rightarrow\mathbb{V}; \ 
    \mu\mapsto \Bigg(\mu\circ{\rm pr}_1^{-1},\mu\circ{\rm pr}_2^{-1},\iint xy\ \mu(dx,dy)\Bigg).$$

The problem to evaluate the left-hand side of \eqref{mainduality} now can be restated in linear programming 
language as follows:
\begin{eqnarray} \label{primal}
   {\rm minimize}: && \iint \ell(x,y)\mu(dx,dy) \\
 \nonumber  {\rm subject \ to}: && A\mu=(P,Q,k) \\
   \nonumber                   && \mu\in{\mathcal K}.
\end{eqnarray}
We will refer to \eqref{primal} as the primal problem. 
To see its equivalence to the original problem, note that square integrability of the marginals implies 
$\mu \in {\mathcal K}$, and that finite Borel measures on $\mathbb{R}^2$ are Radon measures (see, e.g., 
\cite{BauerMI}). Measures $ \mu\in{\mathcal K} $ which satisfy the constraints in the primal problem  
are called feasible. Note that such $\mu$ must be probability measures. 
If feasible solutions to the primal problem exist, it is called consistent. The minimal value of the integral in \eqref{primal} is called the value of the primal problem in optimization theory. In what follows, we will use the more suggestive 
terminology of ``optimal value''.\\

In order to state the corresponding dual problem, we have to compute the adjoint operator $ A^\ast: 
{\mathbb W} \to {\mathbb Y} $ of $ A $.
We have 
\begin{eqnarray*}
 \left\langle A\mu, (f,g,\alpha) \right\rangle&=& \left\langle 
          \left(\mu\circ{\rm pr}_1^{-1},\mu\circ{\rm pr}_2^{-1},\iint xy\ \mu(dx,dy)\right), \left(f,g,\alpha\right)\right\rangle\\
 &=&\int f(x)(\mu\circ{\rm pr}_1^{-1})(dx)+\int g(y)(\mu\circ{\rm pr}_2^{-1})(dy)+\alpha\iint xy\ \mu(dx,dy) \\
 &=&\iint\Big(f(x)+g(y)+\alpha xy\Big)\mu(dx,dy) .
\end{eqnarray*}
Thus we see that 
$$  A^\ast:\mathbb{W}\rightarrow \mathbb{Y}; \
     (f,g,\alpha)\mapsto \Big((x,y)\mapsto f(x)+g(y)+\alpha xy\Big) .$$
The dual problem is now given by 
\begin{eqnarray} \label{dual}
   {\rm maximize}: && \int f(x)P(dx)+\int g(y)Q(dy)+\alpha k \\
 \nonumber  {\rm subject \ to}: && A^\ast(f,g,\alpha)\leq \ell \\
 \nonumber && (f,g,\alpha)\in\mathbb{W}.
\end{eqnarray}
Triplets $ (f,g,\alpha) $ which satisfy the constraints of the dual problem (\ref{dual}) are called feasible.
The dual problem is called consistent, if feasible triplets exist.\\[1em]
 
The central concern of linear programming in infinite dimensional spaces is to give sufficient conditions 
for the optimal values of the primal problem (\ref{primal}) and the dual problem (\ref{dual}) to coincide.
The weak duality theorem (see \cite[Thm.3.1]{AndersonNashBuch}) states:

\begin{proposition}
If the primal problem (\ref{primal}) and the dual problem (\ref{dual}) are consistent,
then the optimal value of the primal problem is larger than the optimal value of the dual problem.
\end{proposition}
The difference between both values is called the duality gap.
So we have to prove that there is no duality gap in the present situation.
The following proposition collects some essential facts from infinite dimensional linear 
programming.
\begin{proposition} \label{kp-stetig}
Assume that the set 
$$  H:=\Big\{(A\mu,\langle\ell,\mu\rangle+r):\ \mu\in{\mathcal K},\ r\geq0\Big\}   $$
is closed in $ \mathbb{V}\times\mathbb{R} $. Then one has: 
\begin{itemize}
\item[(a)] If the primal problem (\ref{primal}) is consistent with finite optimal value, then 
there is no duality gap.
\item[(b)] If the primal problem (\ref{primal}) has finite optimal value, then it is solvable, i.e., 
a minimizing $\mu \in \mathcal{K}$ exists.
\end{itemize}
\end{proposition}
\begin{proof} These are Theorems 3.9 and 3.22 in Anderson and Nash \cite{AndersonNashBuch}, plus
the observation that the proofs of these theorems as well as the proof of 
\cite[Thm.3.3]{AndersonNashBuch} do not use 
the continuity of the functional $ \mu\mapsto\langle\mu,\ell\rangle $. 
\end{proof}
Lemmata \ref{Portmanteau} and \ref{Straffheit} will help us prove that in the present set-up the set $ H $ is closed.

\begin{lemma} \label{Portmanteau} 
For a bounded lower semicontinuous function $ h:\mathbb{R}^2\rightarrow\mathbb{R} $, weak 
convergence  in ${\mathcal R}_+(\mathbb{R}^2) $ of $ \mu_n$ towards $ \mu_\ast$ 
implies
$$  \liminf_{n\rightarrow\infty}\iint h(x,y)\mu_n(dx,dy)\geq \iint h(x,y)\mu_\ast(dx,dy)    .$$
\end{lemma}
\begin{proof}
We may assume without loss of generality that $ h $ is non-negative. We define the functions
$$  g^\ast_m(x,y):=\frac{1}{m}\sum_{k=1}^\infty 1_{h^{-1}((k/m,\infty))}(x,y).$$
Note that this is actually a finite sum of indicators of open sets. 
It then follows that $$ h-\frac{1}{m}\leq g_m^\ast\leq h .$$ 
Using the Portmanteau theorem (see, e.g., \cite[p.24]{BillingsleyConv-alt}) one can deduce from this that 
\begin{eqnarray*}
 \liminf_{n\rightarrow\infty}\iint h(x,y)\mu_n(dx,dy)
 &\geq& \liminf_{n\rightarrow\infty}\iint g^\ast_m(x,y)\mu_n(dx,dy)\\
  &\geq&\iint g^\ast_m(x,y)\mu_\ast(dx,dy)\\ 
   &\geq& \iint h(x,y)\mu_\ast(dx,dy)-1/m
\end{eqnarray*}
Since $ m\in\mathbb{N} $ is arbitrary, the result follows.
\end{proof}

\begin{lemma} \label{Straffheit}
The sequence of weighted measures $ ((x^2+y^2)\vee1)\mu_n  $ is tight if, and only if, the two sequences  
$ (x^2\vee1)(\mu_n\circ{\rm pr}_1^{-1})$ and $ (y^2\vee1)(\mu_n\circ{\rm pr}_2^{-1}) $ both are tight.
\end{lemma}
\begin{proof}
It is obvious that tightness of $ (x^2+y^2)\vee1\ \mu_n $ implies tightness of 
$ (x^2\vee1) (\mu_n\circ{\rm pr}_1^{-1})$ and $ (y^2\vee1)(\mu_n\circ{\rm pr}_2^{-1})$.\\
In order to prove the converse we note that for a prescribed $ \epsilon>0 $ there exists
$ z_0>1 $ with the properties
$$     \int_{[-z_0,z_0]^c}x^2\vee1\ \mu_n\circ{\rm pr}_1^{-1}(dx)<\epsilon \ \ \ {\rm and } \ \ \ 
     \int_{[-z_0,z_0]^c}y^2\vee1\ \mu_n\circ{\rm pr}_2^{-1}(dy)<\epsilon .$$
It follows from this that  
$$   z_0^2\ \mu_n\circ{\rm pr}_1^{-1}([-z_0,z_0]^c)<\epsilon \ \ \ {\rm and} \ \ \ 
     z_0^2\ \mu_n\circ{\rm pr}_2^{-1}([-z_0,z_0]^c)<\epsilon .$$
This implies that
\begin{eqnarray*}
 &&  \iint_{([-z_0,z_0]\times[-z_0,z_0])^c}(x^2+y^2)\vee1\ \mu_n(dx,dy) \\
  &=& \iint_{[-z_0,z_0]\times[-z_0,z_0]^c}x^2\vee1\ \mu_n(dx,dy)
      +\iint_{[-z_0,z_0]\times[-z_0,z_0]^c}y^2\vee1\ \mu_n(dx,dy) \\
  &&+  \iint_{[-z_0,z_0]^c\times[-z_0,z_0]}x^2\vee1\ \mu_n(dx,dy)
      +\iint_{[-z_0,z_0]^c\times[-z_0,z_0]}y^2\vee1\ \mu_n(dx,dy) \\
  &&+ \iint_{[-z_0,z_0]^c\times[-z_0,z_0]^c}x^2\vee1\ \mu_n(dx,dy)
      +\iint_{[-z_0,z_0]^c\times[-z_0,z_0]^c}y^2\vee1\ \mu_n(dx,dy) 
\end{eqnarray*}
\begin{eqnarray*}
  &\leq& z_0^2\mu_n(\mathbb{R}\times[-z_0,z_0]^c)
      +\int_{[-z_0,z_0]^c}y^2\vee1\ \mu_n\circ{\rm pr}_2^{-1}(dy) \\
  &&+  \int_{[-z_0,z_0]^c}x^2\vee1\ \mu_n\circ{\rm pr}_1^{-1}(dx)
      +z_0^2\mu_n([-z_0,z_0]^c\times\mathbb{R}) \\
  &&+ \int_{[-z_0,z_0]^c}x^2\vee1\ \mu_n\circ{\rm pr}_1^{-1}(dx)
      +\int_{[-z_0,z_0]^c}y^2\vee1\ \mu_n\circ{\rm pr}_2^{-1}(dy) \\
  &\leq & 6\epsilon .
\end{eqnarray*}
The tightness of the weighted sequence $ ((x^2+y^2)\vee1) \mu_n$ follows from these considerations. 
\end{proof}
The next lemma will be useful to prove the consistency of the primal problem (\ref{primal}). It concerns the following transportation problems:  
\begin{eqnarray} \label{kleintransport}
    {\rm minimize}: && \int xy\ \mu(dx,dy) \\
 \nonumber  {\rm subject \ to}: && \mu\circ{\rm pr}_1^{-1}=P,\ \mu\circ{\rm pr}_2^{-1}=Q  \\
 \nonumber && \mu\in{\mathcal K}
\end{eqnarray}
and
\begin{eqnarray} \label{grosstransport}
   {\rm maximize}: && \int xy\ \mu(dx,dy) \\
 \nonumber  {\rm subject \ to}: && \mu\circ{\rm pr}_1^{-1}=P,\ \mu\circ{\rm pr}_2^{-1}=Q  \\
 \nonumber && \mu\in{\mathcal K}.
\end{eqnarray}

\begin{lemma}\label{Raender}
There exists an optimal solution $ \mu^- $ for the optimization problem (\ref{kleintransport}) and 
an optimal solution $ \mu^+ $ for the optimization problem (\ref{grosstransport}).
\end{lemma}
\begin{proof} 
It suffices to indicate the proof for (\ref{kleintransport}).
We first note that the product measure $ \mu_0:=P\otimes Q $ satisfies $ \mu_0\circ{\rm pr}_1^{-1}=P $ and 
$ \mu_0\circ{\rm pr}_2^{-1}=Q $. 
Thus the transportation problem (\ref{kleintransport}) is consistent. 
Moreover, we have that 
\begin{eqnarray*}
  \Bigg|\iint xy\ \mu(dx,dy)\Bigg| &\leq& \iint (x^2+y^2)\vee1\ \mu(dx,dy) \\
   &\leq& \int x^2\vee1 \ \mu\circ{\rm pr}_1^{-1}(dx)+ \int y^2\vee1 \ \mu\circ{\rm pr}_2^{-1}(dy)\\
   &\leq& \int x^2\vee1 \ P(dx)+ \int y^2\vee1 \ Q(dy) <\infty .
\end{eqnarray*}
Thus the transportation problem (\ref{kleintransport}) has finite optimal value. To prove the existence of 
an optimal solution, one applies an analogue to Proposition \ref{kp-stetig} (b) (see \cite[Thm. 3.22]{AndersonNashBuch}),
which states that it is sufficient to prove that the set
$$ D := \left\{ \left(\mu \circ {\rm pr}_1^{-1},\ \mu \circ {\rm pr}_2^{-1}, \iint xy\ \mu(dx, dy)\right):\ \mu \in 
\mathcal{K}\right\}$$ is closed in $\sigma(\mathbb{V}, \mathbb{W}).$
Assume that there exists a sequence $ \mu_n $ in $ \mathcal{K} $ with the property that  
$$   \left( \mu_n\circ{\rm pr}_1^{-1},\ \mu_n\circ{\rm pr}_2^{-1}, \iint xy\mu_n(dx,dy) \right) $$ 
converges with respect to $\sigma(\mathbb{V}, \mathbb{W})$ towards $(\tilde{\mu}_1, \tilde{\mu}_2, c)$.  
The convergence of these sequences in the $ \sigma(\mathbb{V},\mathbb{W}) $-topology implies that 
the sequences $ (x^2\vee 1)(\mu_n\circ{\rm pr}_1^{-1}) $ and $ (y^2\vee1)( \mu_n\circ{\rm pr}_2^{-1}) $
are tight. By Lemmata \ref{Straffheit} and \ref{sigmakonv} it follows that the sequences 
$ \nu_n := ((x^2+y^2)\vee1) \mu_n $ and $\mu_n$ are tight. Hence there exist subsequences $\mu_{n_k}$ and
$\nu_{{n_k}}$ converging to $\mu_\ast,\  \nu_\ast = ((x^2 + y^2) \vee 1) \mu_\ast$, respectively.
This means that $\mu_{n_k}$ $\sigma(\mathbb{X}, \mathbb{Y})$-converges to $\mu_\ast$. 
It follows from this that the three sequences 
$$ \mu_n\circ{\rm pr}_1^{-1}, \ \ \mu_n\circ{\rm pr}_2^{-1} 
     \ \ {\rm and} \ \   \iint xy\ \mu_n(dx,dy)   $$
converge respectively towards 
$$ \mu_\ast\circ{\rm pr}_1^{-1}, \ \ \mu_\ast\circ{\rm pr}_2^{-1} 
     \ \ {\rm and} \ \   \iint xy\ \mu_\ast(dx,dy)   ,$$
hence     
$$ \tilde{\mu}_1 = \mu_\ast\circ{\rm pr}_1^{-1}, \ \ \tilde{\mu}_2 = \mu_\ast\circ{\rm pr}_2^{-1} 
     \ \ {\rm and} \ \   c = \iint xy\ \mu_\ast(dx,dy)   .$$     
Thus the set $ D $ is closed and the proof of \eqref{kleintransport} is complete. 
\end{proof}

\begin{proof}[Proof of Theorem \ref{hauptsatz}]

We first prove that the primal problem (\ref{primal}) is consistent if $ k\in\mathbb{R} $ satisfies 
$$ \inf_{\mu\in{\mathcal M}(P,Q)}\iint xy\ \mu(dx,dy)\leq k\leq\sup_{\mu\in{\mathcal M}(P,Q)}\iint xy\ 
\mu(dx,dy) .$$
In Lemma \ref{Raender} we saw that the supremum and infimum are attained in $ \mu^+ $ resp. $ \mu^- $.
So there exists  
$ \lambda\in[0,1] $ with the property that $ \mu_\lambda:=\lambda\mu^++(1-\lambda)\mu^- \in 
\mathcal{M}^1(\mathbb{R}^2)$  
satisfies
$$  \mu_\lambda\circ{\rm pr}_1^{-1}=P,\ \ \mu_\lambda\circ{\rm pr}_2^{-1}=Q \ \  
               {\rm and} \ \  \iint xy\ \mu_\lambda(dx,dy)=k .$$
This proves that the set of feasible solutions for the primal problem (\ref{primal}) is nonempty, 
hence (\ref{primal}) is 
consistent. Moreover, since $ \ell $ is bounded, the objective functional
$$   \mu\mapsto\iint\ell(x,y)\mu(dx,dy)  $$ 
is bounded on the set of feasible solutions. This implies that the 
optimal value must be finite.\\

From Proposition \ref{kp-stetig} it follows that in order to prove 
Theorem \ref{hauptsatz}, it is sufficient to show that the set
\begin{equation}
\label{finalH}
 H:=\Big\{(A\mu,\langle\ell,\mu\rangle+r):\ \mu\in{\mathcal K},\ r\geq0\Big\}   
\end{equation}
is closed in $ \mathbb{V}\times\mathbb{R} $.\\

Consider a sequence of measures $ \mu_n\in{\mathcal K} $ and a sequence of 
real numbers $ r_n \ge 0$ such that the triplets
$$ \left( \mu_n\circ{\rm pr}_1^{-1}, \ \ \mu_n\circ{\rm pr}_2^{-1}, \ \ \iint xy\ \mu_n(dx,dy)\right) $$
converge in  $ \sigma(\mathbb{V},\mathbb{W}) $ towards a triplet $ (\tilde{\mu}_1, \ \tilde{\mu}_2,\ c)\in\mathbb{V} $, and
that 
$$           \iint\ell(x,y)\mu_n(dx,dy)+r_n $$
converges to a real number $ b $. 
We have to find a measure $ \mu_\ast\in{\mathcal K} $ and a real number $ r_\ast\geq0 $ with the 
properties
$$ \mu_\ast\circ{\rm pr}_1^{-1}=\ \tilde{\mu}_1, \ \ \mu_\ast\circ{\rm pr}_2^{-1}=\tilde{\mu}_2, \ \ \iint xy\mu_\ast(dx,dy)=c 
     \ \ {\rm and} \ \     \iint\ell(x,y)\mu_\ast(dx,dy)+r_\ast=b .$$
The convergence of these sequences in the $ \sigma(\mathbb{V},\mathbb{W}) $-topology implies that 
the sequences $ (x^2\vee 1)(\mu_n\circ{\rm pr}_1^{-1}) $ and $ (y^2\vee1)( \mu_n\circ{\rm pr}_2^{-1}) $
are tight. By Lemmata \ref{Straffheit} and \ref{sigmakonv} it follows that the sequences 
$ \nu_n := ((x^2+y^2)\vee1) \mu_n $ and $\mu_n$ are tight. Hence there exist subsequences $\mu_{n_k}$ and
$\nu_{{n_k}}$ converging to $\mu_\ast,\  \nu_\ast = ((x^2 + y^2) \vee 1) \mu_\ast$, respectively.
This means that $\mu_{n_k}$ $\sigma(\mathbb{X}, \mathbb{Y})$-converges to $\mu_\ast$. 
It follows from this that the three sequences 
$$ \mu_n\circ{\rm pr}_1^{-1}, \ \ \mu_n\circ{\rm pr}_2^{-1} 
     \ \ {\rm and} \ \   \iint xy\ \mu_n(dx,dy)   $$
converge respectively towards 
$$ \mu_\ast\circ{\rm pr}_1^{-1}, \ \ \mu_\ast\circ{\rm pr}_2^{-1} 
     \ \ {\rm and} \ \   \iint xy\ \mu_\ast(dx,dy)   ,$$
hence     
$$ \tilde{\mu}_1 = \mu_\ast\circ{\rm pr}_1^{-1}, \ \ \tilde{\mu}_2 = \mu_\ast\circ{\rm pr}_2^{-1} 
     \ \ {\rm and} \ \   c = \iint xy\ \mu_\ast(dx,dy)   .$$

Moreover, since $ \ell $ is bounded and lower semicontinuous,  by Lemma \ref{Portmanteau}  one has 
$$   \liminf_{k\rightarrow\infty}\iint\ell(x,y)\mu_{n_k}(dx,dy)\geq\iint\ell(x,y)\mu_\ast(dx,dy) > - \infty .$$
By further thinning out the subsequence $ \mu_{n_k} $, we may assume without loss of generality that 
$$   \lim_{k\rightarrow\infty}\iint\ell(x,y)\mu_{n_k}(dx,dy)\geq\iint\ell(x,y)\mu_\ast(dx,dy) > - \infty . $$
Together with the convergence of the sequence 
$$  \iint\ell(x,y)\mu_n(dx,dy)+r_n,$$ 
this implies that the sequence $ r_{n_k} $ converges to a non-negative real number $ r_\infty $. 
If we define
$$ r_\ast:= r_\infty+\lim_{k\rightarrow\infty}\iint\ell(x,y)\mu_{n_k}(dx,dy)-\iint\ell(x,y)\mu_\ast(dx,dy),  $$
it follows that
\begin{eqnarray*}
        \lim_{n\rightarrow\infty}\Bigg(\iint\ell(x,y)\mu_n(dx,dy)+r_n\Bigg) 
  &=&    \lim_{k \rightarrow\infty}\Bigg(\iint\ell(x,y)\mu_{n_k}(dx,dy)+r_{n_k}\Bigg) \\
  &=&    \iint\ell(x,y)\mu_\ast(dx,dy)+r_\ast . 
\end{eqnarray*}
Since by definition $ r_\ast $ is non-negative, it follows that the set $ H $ is closed.
This completes the proof of Theorem \ref{hauptsatz}.
\end{proof}

\bibliographystyle{amsplain}
\bibliography{BriceMichael}

\providecommand{\bysame}{\leavevmode\hbox to3em{\hrulefill}\thinspace}
\providecommand{\MR}{\relax\ifhmode\unskip\space\fi MR }
\providecommand{\MRhref}[2]{%
  \href{http://www.ams.org/mathscinet-getitem?mr=#1}{#2}
}
\providecommand{\href}[2]{#2}
\begin{thebibliography}{10}

\bibitem{AndersonNashBuch}
Edward~J. Anderson and Peter Nash, \emph{Linear programming in
  infinite-dimensional spaces}, Wiley-Interscience Series in Discrete
  Mathematics and Optimization, John Wiley \& Sons Ltd., Chichester, 1987,
  Theory and applications, A Wiley-Interscience Publication. \MR{MR893179
  (88f:90180)}

\bibitem{BauerMI}
Heinz Bauer, \emph{Measure and integration theory}, de Gruyter Studies in
  Mathematics, vol.~26, Walter de Gruyter \& Co., Berlin, 2001, Translated from
  the German by Robert B. Burckel. \MR{MR1897176 (2003a:28001)}

\bibitem{BillingsleyConv-alt}
Patrick Billingsley, \emph{Convergence of probability measures}, John Wiley \&
  Sons Inc., New York, 1968. \MR{MR0233396 (38 \#1718)}

\bibitem{EmbrechtsHoingJuri}
Paul Embrechts, Andrea H{\"o}ing, and Alessandro Juri, \emph{Using copulae to
  bound the value-at-risk for functions of dependent risks}, Finance Stoch.
  \textbf{7} (2003), no.~2, 145--167. \MR{MR1968943 (2004c:91048)}

\bibitem{FollmerSchiedBuch}
Hans F{\"o}llmer and Alexander Schied, \emph{Stochastic finance}, extended ed.,
  de Gruyter Studies in Mathematics, vol.~27, Walter de Gruyter \& Co., Berlin,
  2004, An introduction in discrete time. \MR{MR2169807 (2006d:91002)}

\bibitem{FrankNelsenSchweizer}
M.~J. Frank, R.~B. Nelsen, and B.~Schweizer, \emph{Best-possible bounds for the
  distribution of a sum---a problem of {K}olmogorov}, Probab. Theory Related
  Fields \textbf{74} (1987), no.~2, 199--211. \MR{MR871251 (88f:60030)}

\bibitem{JimenezRodriguez}
P.~Jim{\'e}nez~Guerra and B.~Rodr{\'{\i}}guez-Salinas, \emph{A general solution
  of the {M}onge-{K}antorovich mass-transfer problem}, J. Math. Anal. Appl.
  \textbf{202} (1996), no.~2, 492--510. \MR{MR1406244 (97h:49006)}

\bibitem{Makarov1981-engl}
G.~D. Makarov, \emph{Estimates for the distribution function of the sum of two
  random variables with given marginal distributions}, Theory Probab. Appl.
  \textbf{26} (1981), no.~4, 803--806. \MR{MR636775 (83c:60029)}

\bibitem{McNeilFreyEmbrechts}
Alexander~J. McNeil, R{\"u}diger Frey, and Paul Embrechts, \emph{Quantitative
  risk management}, Princeton Series in Finance, Princeton University Press,
  Princeton, NJ, 2005, Concepts, techniques and tools. \MR{MR2175089
  (2006d:91005)}

\bibitem{MikoschExtremes2006}
Thomas Mikosch, \emph{Copulas: tales and facts}, Extremes \textbf{9} (2006),
  no.~1, 3--20. \MR{MR2327842}

\bibitem{RaRu1}
Svetlozar~T. Rachev and Ludger R{\"u}schendorf, \emph{Mass transportation
  problems. {V}ol. {I}}, Probability and its Applications (New York),
  Springer-Verlag, New York, 1998, Theory. \MR{MR1619170 (99k:28006)}

\bibitem{RaRu2}
\bysame, \emph{Mass transportation problems. {V}ol. {II}}, Probability and its
  Applications (New York), Springer-Verlag, New York, 1998, Applications.
  \MR{MR1619171 (99k:28007)}

\bibitem{Romeike-etal}
Frank Romeike, Matthias M\"uller-Reichart, and Thorsten Hein, \emph{Die
  {A}ssekuranz am {S}cheideweg -- {E}rgebnisse der ersten {B}enchmark-{S}tudie
  zu {S}olvency {II}}, Zeitschrift f\"ur Versicherungswesen (2006), no.~10,
  316--321.

\bibitem{RuschendorfAdvApplProb1982}
Ludger R{\"u}schendorf, \emph{Random variables with maximum sums}, Adv. in
  Appl. Probab. \textbf{14} (1982), no.~3, 623--632. \MR{MR665297 (83j:60021)}

\bibitem{VillaniOldNew}
C{\'e}dric Villani, \emph{Optimal transport}, Grundlehren der Mathematischen
  Wissenschaften [Fundamental Principles of Mathematical Sciences], vol. 338,
  Springer-Verlag, Berlin, 2009, Old and new. \MR{MR2459454}

\end{thebibliography}

\end{document}